\documentclass[reqno,a4paper]{amsart}
\usepackage[
    style=numeric-comp,sorting=nty,
    sortcites=true,doi=false,url=false,giveninits=true]{biblatex}
    \renewbibmacro{in:}{}

\usepackage{amsmath}
\usepackage{amsthm}
\usepackage{amssymb}
\usepackage{float}
\usepackage{enumitem}
\usepackage{color}

\numberwithin{equation}{section}

\theoremstyle{plain}
\newtheorem{theorem}{Theorem}[section]
\newtheorem{lemma}[theorem]{Lemma}
\newtheorem{corollary}[theorem]{Corollary}
\newtheorem{proposition}[theorem]{Proposition}
\newtheorem{definition}[theorem]{Definition}
\theoremstyle{remark}
\newtheorem{remark}[theorem]{Remark}

\newcommand\Pone{\textrm{P}_{\textrm{I}} }
\newcommand\Ptwo{\textrm{P}_{\textrm{II}} }
\newcommand\Ptwomu{\textrm{P}_{\textrm{II}}(\mu) }

\newcommand{\orcidauthorA}{0000-0001-7504-4444}
\newcommand{\orcidauthorB}{0000-0003-1354-1271}

\title{On the Perturbed Second Painlev\'{e} Equation}
\author{Joshua Holroyd}
\thanks{JH's ORCID ID is \orcidauthorB.
}
\address{School of Mathematics and Statistics F07, University of Sydney, NSW 2006 Australia}
\email{j.holroyd@maths.usyd.edu.au}
\author{Nalini Joshi}
\thanks{NJ's ORCID ID is \orcidauthorA. NJ's research was supported by an
  Australian Research Council Discovery Project \#DP210100129.
}
\address{School of Mathematics and Statistics F07, University of Sydney, NSW 2006 Australia}
\email{nalini.joshi@sydney.edu.au}

\date{\today}

\begin{document}
\maketitle
{\centering\footnotesize \emph{Dedicated to Sir Michael V. Berry on the occasion of his $\textit{80}^{\textit{th}}$ birthday.}\par}
\begin{abstract}
We consider a perturbed version of the second Painlev\'{e} equation ($\Ptwo$), which arises in applications, and show that it possesses solutions analogous to the celebrated Hastings-McLeod and tritronqu\'ee solutions of $\Ptwo$. The Hastings-McLeod-type solution of the perturbed equation is holomorphic, real-valued and positive on the whole real-line, while the tritronqu\'ee-type solution is holomorphic in a large sector of the complex plane. These properties also characterise the corresponding solutions of $\Ptwo$ and are surprising because the perturbed equation does not possess additional distinctive properties that characterise $\Ptwo$, particularly the Painlev\'e property.
\end{abstract}
\bigskip


\section{Introduction}

The second Painlev\'{e} Equation ($\Ptwo$) has a rich history as a mathematical model, arising in a very wide range of applications (see, for example, \cites{clarkson2006painleve,forrester2015painleve,tracy1994level,tracy1999random}). In many of these applications, the particular case of $\Ptwo$ that arises is the case 
$\alpha=0$ of 
\begin{equation*} \label{eq:p2}
    \Ptwo :\qquad y''(x)=2y(x)^3+x y(x)+\alpha,
\end{equation*}
where primes denote differentiation in $x$ and $\alpha$ is a given parameter.
In this paper, we consider a perturbed version of $\Ptwo$ with $\alpha=0$:
\begin{equation*} \label{eq:p2gen}
    \Ptwomu :\qquad y''(x)=2y(x)^3+x^\mu y(x),\qquad y:\mathbb{C}\rightarrow\mathbb{C},\qquad \mu\in\mathbb{Z}^+.
\end{equation*}
The main results of this paper show that there exists (i) a holomorphic, real-valued, positive solution of $\Ptwomu$ on the real-line $x\in\mathbb R$; and (ii) a solution that is asymptotic to a power-series expansion in $x$, as $|x|\to\infty$, in a sector of angular extent $4\pi/(\mu+2)$ in $\mathbb C$. These properties mirror two well-known solutions of $\Ptwo$, called the Hastings-McLeod solution \cite{hastings1980boundary} and Boutroux's tritronqu\'ee solution \cite{boutroux1913recherches} respectively, which are distinguished by their global properties.

$\Ptwo$ belongs to a class of six ordinary differential equations (ODEs), called the Painlev\'e equations, whose solutions are known to have only poles as movable singularities \cites{conte2012painleve,i:56}. Every non-rational solution of $\Ptwo$ is known to possess an infinite number of movable poles \cite{howes2014joshi}, which are typically distributed across the interior of each sector $S_n$ of the complex plane given by \cite{boutroux1913recherches} \begin{equation}\label{eq:sectors}
S_n=\bigl\{x\in\mathbb C \bigm| (n-1)\pi/3\leq\arg{x}\le n\pi/3\bigr\},\ n\in\mathbb Z, \mod{6}.
\end{equation}
(Here, the word ``movable'' refers to the location of each pole depending on initial conditions.) Therefore, solutions that are free of movable poles across one or more sectors $S_n$ are uncommon, while those free of poles in a domain containing the real line are of particular interest for physical applications. The tritronqu\'ee solution is free of poles in a wide sector of the complex plane, given by four contiguous sectors $S_n$, while the Hastings-McLeod solution is free of poles and analytic on the whole real-line. 

The Hastings-McLeod solution now plays a critical role in modern random matrix theory, due to its appearance in the  the Tracy-Widom distribution, describing the distribution of the largest eigenvalue of Gaussian unitary ensembles of random matrices, as the size of each matrix approaches infinity. The Hastings-McLeod solution also arises in the analogous distribution functions for the Gaussian orthogonal and symplectic ensembles  \cite{tracy1994level,clarkson2006painleve,Deift07}.

There are many more physical models in which $\Ptwo$ arises. We mention, in particular, fluid dynamics \cite{fokas1981linearization,gromak2008painleve,clarkson2006painleve,olver2011numerical,ablowitz1977exact,claeys2010painleve}, mathematical physics \cite{periwal1990exactly,seiberg2005flux,tracy1999random,schiappa2014resurgence,forrester2015painleve}, and electrodynamics \cite{bass2010electrical,bass1964irreversible}. In particular, the second Painlev\'{e} Equation arises in a model of steady one-dimensional two-ion electro-diffusion \cite{bass1964irreversible}. A review may be found in \cite{conte2012painleve}. In connection with plasma physics, de Boer and Ludford \cite{de1975spherical} posed the question of the existence of a solution of a generalised form of $\Ptwo$ with specific boundary conditions. Hastings and McLeod \cite{hastings1980boundary} not only showed the existence of a solution to this problem, but they also proved it to be unique; this is the now-famous Hastings-McLeod solution of $\Ptwo$. It is natural in many of the physical contexts to consider small perturbations of the model of interest. We describe such a perturbation of Bass' electro-diffusion model in Appendix \ref{section:app} and show how $\Ptwomu$ arises naturally.

A perturbed version of the first Painlev\'e equation has been previously considered \cites{joshi2003PImu,oldedaalhuis22}. Tronqu\'ee and tritronqu\'ee solutions of perturbed $\Pone$ were shown to exist (and be unique in the tritronqu\'ee case) in \cites{joshi2003PImu}. Meanwhile, \cites{oldedaalhuis22} studies the full exponentially-improved asymptotic expansion corresponding to a formal series solution of perturbed $\Pone$. These tronqu\'ee solutions are uniquely characterised by the Stokes multiplier, which is an arbitrary constant multiplying the exponentially small terms. Asymptotic approximations for the locations of singularities along boundaries of validity are deduced in terms of the Stokes multipliers \cites{oldedaalhuis22}.

In the case of $\Ptwomu$, tronqu\'ee solutions are associated with a formal power-series solution $y_f\sim(-x^\mu/2)^{1/2}$ as $|x|\rightarrow\infty$ (given in full in Section \ref{section:TT}), we mention that the (singly) exponentially small perturbation is given by
\begin{equation*}
    y\sim y_f+kx^{-\mu/4}\exp{\left(\frac{2}{\mu+2}\left(-2x^{\mu+2}\right)^{1/2}+\textit{o}(1)\right)},
\end{equation*}
where the square root branch depends on the sector of the complex plane, and $k$ is the free constant multiplying exponentially small behaviour. The tritronqu\'ee solution, which we study in Section \ref{section:TT}, corresponds to the unique $k=0$ case of this expansion.

\subsection{Main results} We identify a tritronqu\'ee solution of $\Ptwomu$, through asymptotic expansions of solutions as $|x|\to\infty$; showing that a true solution asymptotic to such a series in the sector $|\arg(x)|\leq 2\pi/(\mu+2)$ exists and is unique. These results are given in Propositions \ref{prop:tronquee} and \ref{prop:tritronquee}. We also consider a boundary-value problem on the real-line; as $x\to-\infty$, the solution is asymptotic to a power-series expansion, and as $x\to+\infty$, the solution decays to zero exponentially fast. We show that this problem has a unique solution in the case of $\mu=2k+1$, $k\in\mathbb N$, which generalises the Hastings-McLeod solution. We give this result in Theorem \ref{thrm:HM}.

\subsection{Background}
The Painlev\'e equations were found in a search for ODEs that defined new higher transcendental functions a century ago. Here the word ``new" refers to these functions not being able to be expressed in terms of any classical operations applied to previously known such functions (we refer to definitions of these terms in \cite{umemura2007invitation}). Both the Hastings-McLeod and tritronqu\'ee solutions are examples of such transcendental functions. 

For special values of the parameter $\alpha$, $\Ptwo{}$ admits hierarchies of rational solutions \cite{airault1979rational,albrecht1996algorithms,fukutani2000special,clarkson2003second} and special solutions related to the classical Airy functions \cite{gambier1910}. Beyond special solutions, transcendental solutions exist that vanish as $x\to+\infty$ on the real line, providing another relation to the Airy functions in the case when $\alpha=0$. The Hastings-McLeod solution first attracted attention in the asymptotic analysis of such solutions. 

The asymptotic analysis of the Painlev\'e equations was initiated by Boutroux \cite{boutroux1913recherches}, who also briefly considered the generalised equation $\Ptwomu$. Boutroux showed that the transformed version of $\Ptwo$
\begin{equation*} \label{eq:bp2}
    u''(z)=2u^3(z)+u(z)+\frac{2\alpha}{3z}-\frac{u'(z)}{z}+\frac{u(z)}{9z^2},
\end{equation*}
is more amenable to asymptotic analysis, where
\begin{equation*}
    y(x)=x^{1/2}u(z)\qquad\text{and}\qquad z=\frac{2}{3}x^{3/2}.
\end{equation*}
The limit $|x|\to\infty$ now becomes $|z|\to\infty$, and direct observation shows that a (Jacobi) elliptic function describes the leading-order behaviour, solving the ODE
\begin{equation*}
    w''(z)=2w^3(z)+w(z).
\end{equation*}
Boutroux went on to describe how lines of movable poles curve around to approach the boundaries of the sectors $S_n$ described by Equation \eqref{eq:sectors}.

\subsection{Outline}
In Section \ref{section:asymptotics}, we describe decaying solutions of $\Ptwomu$ as $|x|\rightarrow\infty$ on the real-line. We prove the existence and uniqueness of the analogue of Boutroux's tritronqu\'ee solution for $\Ptwomu$ and its asymptotic validity in a broad sector of the complex plane in Section \ref{section:TT}. In Section \ref{section:HM}, we prove the existence and uniqueness of a natural generalisation of the  Hastings-McLeod solution. The paper's results are summarised in Section \ref{sec:conc}. Finally, we show how $\Ptwomu$ arises in a physical setting in Appendix \ref{section:app}.

\section{Asymptotics of decaying solutions on the real-line}
\label{section:asymptotics}

Since $\mu>0$, the decaying solutions of $\Ptwomu$ satisfy the leading order equation
\[
y''(x)\sim x^\mu y(x),
\]
as $x\to\pm\infty$ on the real-line. We analyse the asymptotic behaviours of $y(x)$, which satisfy this leading-order equation. 

The linear ODE $\phi''(x)=x^\mu \phi(x)$ is a generalised Airy equation and is satisfied by two linearly independent functions
\begin{equation} \label{eq:genAi}
    \begin{split}
        f(x)=\sqrt{\frac{2x}{\pi\alpha}}K_{1/(2\alpha)}\left(\frac{x^\alpha}{\alpha}\right)\qquad\text{and}\qquad g(x)=\sqrt{\frac{2\pi x}{\alpha}}I_{1/(2\alpha)}\left(\frac{x^\alpha}{\alpha}\right),
    \end{split}
\end{equation}
where $\alpha=(\mu+2)/2$, and $I_\nu$ and $K_\nu$ are the modified Bessel functions of the first and second kind respectively. The functions $f(x)$ and $g(x)$ are analogues of the well known Airy functions Ai$(x)$ and Bi$(x)$, reducing to scaled versions of these when $\alpha=3/2$ (or equivalently $\mu=1$).
 
Recalling that $\mu$ is positive and real, we find that
\begin{equation}\label{eq:expDecay}
    y(x)\sim kf(x)=kx^{-\mu/4}\exp\left(-\frac{2}{\mu+2}x^{(\mu+2)/2}+\textit{o}(1)\right), \ {\rm as}\ x\to+\infty,
\end{equation}
for some arbitrary constant $k$. As is well known for Airy functions, a solution $y(x)$ of $\Ptwomu$ that decays to zero as $|x|\rightarrow\infty$ may not be exponentially small in other directions. In particular, for every odd integer $\mu$, the negative real semi-axis is an anti-Stokes line \cite{bender2013advanced}, meaning that the exponential contribution is purely oscillatory as $|x|\rightarrow\infty$. Given an odd positive integer $\mu>1$, standard asymptotic techniques lead to 
 \begin{equation}\label{eq:expOscillatory}
     y(x)\sim c_1|x|^{-\mu/4}\sin\left(\frac{2}{\mu+2}|x|^{(\mu+2)/2}-c_2+\textit{o}(1)\right)\qquad\text{as}\qquad x\rightarrow-\infty,
\end{equation}
for arbitrary constants $c_1$ and $c_2$.
In the case $\mu=1$, more detailed results were described by Hastings and McLeod \cite{hastings1980boundary}.


\section{Tritronqu\'{e}e Solutions}\label{section:TT}
The main object of study in this section is an asymptotic series, which solves $\Ptwomu$ given in Equation \eqref{eq:yf}. We prove the existence and uniqueness of a true solution that is asymptotic to this series for $|\arg(x)|\leq 2\pi/(\mu+2)$ as $|x|\rightarrow\infty$. This solution corresponds to the tritronqu\'ee solution in the case $\mu=1$. 

We refer the reader to Boutroux \cite{boutroux1913recherches} for detailed descriptions of the solutions he called {\em int\'egrales tronqu\'ees}, which are pole-free (for sufficiently large $x$) within two adjacent sectors of the form $S_n$ defined in Equation \eqref{eq:sectors}. The region $S_n\cup S_{n+1}$ ($n\in\mathbb Z$ mod 6) described by such a pair of sectors has a bisector given by a ray $\arg x=n\pi/3$. The term tronqu\'ee arises from the fact that any line of poles that originally lies on such a ray must be truncated as $|x|$ increases (due to the pole-free nature of the solution).

There are three such rays in four contiguous sectors $S_n\cup S_{n+1}\cup S_{n+2}\cup S_{n+3}$ ($n\in\mathbb Z$ mod 6). Boutroux showed that, for each choice of such a region, there exists a unique solution he called {\em tritronqu\'ee} (triply-truncated), which is asymptotically pole-free in the region. For $\Ptwomu$, these special rays are given by arg$(x)=n\pi/(\mu+2)$ for $n\in\mathbb{Z}$. 

We deduce an asymptotic series expansion of a tritronqu\'ee solution of $\Ptwomu$ in Proposition \ref{thrm:formalseries}. In Propositions \ref{prop:tronquee} and \ref{prop:tritronquee} we prove the existence and uniqueness of such a solution for arbitrary $\mu\in\mathbb{Z}^+$. 
The proof of existence relies on Wasow's theorem. The further step required for uniqueness is analogous to the proof given by Joshi and Kitaev in \cite{joshi2001boutroux} for the first Painlev\'e equation. We start by giving a  generalisation of Boutroux's transformation of variables.

\begin{definition}[Boutroux variables] \label{def:boutrouxvars}
Let $y(x)$ be a solution of $\Ptwomu$. We define new variables $u(z)$ and $z$ by
\begin{equation*}
    y(x)=x^{\mu/2}u(z)\qquad\text{and}\qquad z=\frac{2}{\mu+2}x^{(\mu+2)/2}.
\end{equation*}
\end{definition}
Substitution into $\Ptwomu$ shows that the function $u(z)$ satisfies
\begin{equation} \label{eq:boutrouxODE}
    u''(z)=2u(z)^3+u(z)-\frac{3\mu}{\mu+2}\frac{u'(z)}{z}-\frac{\mu(\mu-2)}{(\mu+2)^2}\frac{u(z)}{z^2}.
\end{equation}
In the following proposition, we consider its asymptotic expansion as $|z|\to\infty$. 
\begin{proposition} \label{thrm:formalseries}
    The formal series
    \begin{equation*}
        u_f(z)=\sum_{n=0}^{\infty} \frac{a_{n}}{z^{2n}},
    \end{equation*}
    where $a_0=i/\sqrt{2}$, $a_1=-a_0\frac{\mu(\mu-2)}{2(\mu+2)^2}$, and the coefficients $a_n$, for all $n>1$, are given by the recurrence relation
    \begin{equation*}
        \begin{split}
            a_n&=\sigma(n)\,a_{n-1}+a_0\sum_{k=1}^{n-1}a_ka_{n-k}+\sum_{j=1}^{n-1}\sum_{k=0}^{n-j}a_ja_ka_{n-j-k},\\
            \text{where}\qquad\sigma(n)&=-(2n-1)(n-1)+\frac{3\mu}{\mu+2}(n-1)-\frac{\mu(\mu-2)}{2(\mu+2)^2},
        \end{split}
    \end{equation*}
    is a formal solution of Equation (\ref{eq:boutrouxODE}).
\end{proposition}
\begin{proof}
    Substitution of the given series into Equation (\ref{eq:boutrouxODE}).
\end{proof}
This gives rise to a formal series solution of $\Ptwomu$ in the original variables:
\begin{equation} \label{eq:yf}
    y_f(x)=\sqrt{x^{\mu}}\sum_{n=0}^{\infty}\frac{\left(\frac{\mu+2}{2}\right)^{2n}a_n}{x^{(\mu+2)n}},
\end{equation}
with coefficients $a_n$ described in Proposition \ref{thrm:formalseries}.

\begin{remark}\label{rem:reality}
Recalling that $a_0$ is purely imaginary, we note that for $x$ on the negative real semi-axis, $y_f(x)$ may only be real-valued for the case when $\mu$ is odd. 
\end{remark}

\subsection{Existence}
We now prove the existence of true solutions asymptotic to the given power-series in a given sector of angular width up to $2\pi/(\mu+2)$ in the following proposition.
\begin{proposition}[tronqu\'ee solutions] \label{prop:tronquee}
    Given $\mu\in\mathbb{Z}^+$ and a sector $\mathcal{S}$ of angle less than $2\pi/(\mu+2)$, there exists a solution of $\Ptwomu$ whose asymptotic behaviour as $|x|\rightarrow\infty$ in $\mathcal{S}$ is given by the asymptotic series (\ref{eq:yf}).
\end{proposition}
\begin{proof}
    Let $y(x)$ be a solution of $\Ptwomu$. We use the change of variables given in Definition \ref{def:boutrouxvars} and further define
    \begin{equation*}
        w(z)=u(z)-a_0=u(z)-i/\sqrt{2}.
    \end{equation*}
    Then a straightforward calculation shows that the function $w(z)$ satisfies
    \begin{equation} \label{eq:wODE}
        w''(z)=2w(z)^3+3i\sqrt{2}w(z)^2-2w(z)-\frac{3\mu}{\mu+2}\frac{w'(z)}{z}-\frac{\mu(\mu-2)}{(\mu+2)^2}\frac{(w(z)+i/\sqrt{2})}{z^2},
    \end{equation}
    and this ODE admits the formal power-series solution
    \begin{equation*}
        w_f(z)=\sum_{n=1}^{\infty} \frac{a_{n}}{z^{2n}}.
    \end{equation*}
    Now, by defining variables $w_1=w$ and $w_2=w'$, Equation (\ref{eq:wODE}) becomes the system
    \begin{equation*}
        \begin{split}
            w_1'&=w_2,\\
   w_2'&=2w_1^3+3i\sqrt{2}w_1^2-2w_1-\frac{3\mu}{\mu+2}\frac{w_2}{z}-\frac{\mu(\mu-2)}{(\mu+2)^2}\frac{(w_1+i/\sqrt{2})}{z^2}.
        \end{split}
    \end{equation*}
    Let $f_1=w_2$ and $f_2$ be the right side of the second equation. This system is, therefore, in the form
    \begin{equation*}
        z^{-q}\frac{d\boldsymbol{w}}{dz}=\boldsymbol{f}(z,w_1,w_2),
    \end{equation*}
    where $\boldsymbol{w}=(w_1, w_2)^T$ $\boldsymbol{f}=(f_1, f_2)^T$ and $q=0$. The Jacobian of $\boldsymbol{f}$, evaluated at $\boldsymbol{w}=\boldsymbol{0}$ and $|z|\rightarrow\infty$, has nonzero eigenvalues $\lambda_{\pm}=\pm i\sqrt{2}$. These results allow us to apply Wasow's theorem  \cite[Theorem 12.1]{wasow2018asymptotic}, showing the existence of a true solution with the desired asymptotic behaviour.
\end{proof}
\begin{definition}
    The solutions defined by Proposition \ref{prop:tronquee} are called tronqu\'{e}e solutions.
\end{definition}
\begin{corollary}
    Let $\mu\in\mathbb{Z}^+$ and $\mathcal{S}\in\mathbb C$ be a given sector of angle less than $2\pi/(\mu+2)$. Suppose $y(x)$ is a tronqu\'{e}e solution of $\Ptwomu$ in $\mathcal S$. Then the $k$th derivative, $y^{(k)}(x)$, for $k\in\{0,1,2, \ldots\}$, has asymptotic behaviour given by the $k$th term-by-term differentiation of series (\ref{eq:yf}), i.e., $y_f^{(k)}(x)$, as $|x|\rightarrow\infty$ in $\mathcal{S}$. 
\end{corollary}
\begin{proof}
    The asymptotic behaviour of $y(x)$ implies that it is holomorphic in $\mathcal{S}$ for sufficiently large $x$. The proof follows from Wasow's theorem \cite[Theorem 8.8]{wasow2018asymptotic}.
\end{proof}

\subsection{Uniqueness} In the following proposition, we prove that there is a unique true solution asymptotic to the given power-series in $|\text{arg}(x)|\leq2\pi/(\mu+2)$.

\begin{proposition}[tritronqu\'ee solution]\label{prop:tritronquee}
    Given $\mu\in\mathbb{Z}^+$, there exists a unique solution $Y(x)$ of $\Ptwomu$ that has the asymptotic expansion $y_f(x)$ in the sector $|\arg(x)|\leq2\pi/(\mu+2)$ as $|x|\rightarrow\infty$.
\end{proposition}
\begin{proof}
    Let $0<\epsilon\in\mathbb{R}$. Let $y_1(x)$ and $y_2(x)$ be tronqu\'{e}e solutions as defined by Proposition \ref{prop:tronquee} in the respective sectors
    \begin{equation*}
        S_1:-\epsilon/2<\text{arg}(x)<\theta-\epsilon \qquad\text{and}\qquad S_2:-\theta+\epsilon<\text{arg}(x)<\epsilon/2,
    \end{equation*}
    where $\theta=2\pi/(\mu+2)$. Consider solutions $y_1(x)$ and $y_2(x)$ in the overlapping sector of angular width $\epsilon$, centered about the positive real axis. Since the two tronqu\'{e}e solutions have the same asymptotic expansion, we get
    \begin{equation} \label{eq:asympv}
        v(x):=y_1(x)-y_2(x)\underset{|\text{arg}(x)|<\epsilon/2}{\underset{|x|\rightarrow\infty}{=}}\textit{o}\left(\frac{\sqrt{x^\mu}}{x^{(\mu+2)n}}\right),
    \end{equation}
    for every $n\in\{0,1,2, \ldots\}$. Furthermore, we have $v''(x)=y_1''(x)-y_2''(x)$, and by applying $\Ptwomu$, we see that $v(x)$ satisfies the linear ODE
    \begin{equation*}
        v''(x)=(2y_1(x)^2+2y_1(x)y_2(x)+2y_2(x)^2+x^\mu)v=(-f(x)+g(x))v,
    \end{equation*}
    where
    \begin{equation*}
        f(x)=2x^\mu\qquad\text{and}\qquad g(x)=2y_1(x)^2+2y_1(x)y_2(x)+2y_2(x)^2+3x^\mu.
    \end{equation*}
    The functions $f$ and $g$ are holomorphic for sufficiently large $x$ in the sector $S_1\cap S_2$. It follows from asymptotic series (\ref{eq:yf}) that $g$ is $\mathcal{O}(1/x^2)$, and thus the integral
    \begin{equation*}
        \int_{x}^{\infty}\left|f^{-1/4}\left(f^{-1/4}\right)''-gf^{-1/2}\right|dx,
    \end{equation*}
    taken along the real axis converges. Therefore, we may apply Olver's theorem \cite[Theorem 2.2]{olver1997asymptotics} to obtain
    \begin{align*}
        v(x)=&c_1f^{-1/4}\exp\left(i\int_{x}^{\infty}f^{1/2}dx\right)\bigl(1+\textit{o}(1)\bigr)\\
        &+c_2f^{-1/4}\exp\left(-i\int_{x}^{\infty}f^{1/2}dx\right)\bigl(1+\textit{o}(1)\bigr),
    \end{align*}
    where $c_1$ and $c_2$ are some constants. For real $x$, this contradicts the asymptotic behaviour (\ref{eq:asympv}) unless $c_1=c_2=0$. So it follows that $y_1(x)=y_2(x)$ for sufficiently large real $x$.
\end{proof}
\begin{definition}
    The solution defined by Proposition \ref{prop:tritronquee} is called the tritronqu\'ee solution $Y(x)$.
\end{definition}

\begin{remark} Note that the sector in which $Y$ is asymptotically free of poles is symmetric around the positive real semi-axis.
\end{remark}

Given $Y(x)$, defined above, we may generate other tritronqu\'ee solutions by using the discrete symmetry of $\Ptwomu$. The functions
\begin{equation}\label{eq:symmetry}
    Y_n(x)=\exp{\left(-\frac{2\pi\,i\,n}{\mu+2}\right)}Y\left(\exp{\left(-\frac{2\pi\,i\,n}{\mu+2}\right)}x\right),
\end{equation}
for $n\in\mathbb{Z}$, also satisfy $\Ptwomu$ and are tritronqu\'ee in rotated sectors of width $4\pi/(\mu+2)$. This transformation gives a tritronqu\'ee solution corresponding to each ray of angle $2\pi n/(\mu+2)$.


\section{The Hastings-McLeod-Type Solution}\label{section:HM}

In this section, we define the generalised Hastings-McLeod solution via a boundary-value problem on the real-line. We then prove that the existence of this solution fails when $\mu$ is even while it exists and is unique when $\mu$ is odd.

\begin{definition}[Hastings-McLeod-type solution]\label{def:HM}
    Given $\mu\in\mathbb{Z}^+$, for $x\in\mathbb{R}$, we call a solution $y(x)$ of $\Ptwomu$ a Hastings-McLeod-type solution if and only if $y(x)\sim y_f(x)$ as $x\rightarrow-\infty$ (see Equation (\ref{eq:yf})), and $y(x)\sim kf(x)$ as $x\rightarrow+\infty$ (see Equations (\ref{eq:genAi})) for some $k\in\mathbb{R}^+$.
\end{definition}
We now prove the main result of this section.
\begin{theorem}\label{thrm:HM}
    For every odd integer $\mu>1$, there exists a unique, Hastings-McLeod-type solution $y(x)$ of $\Ptwomu$. On the other hand, for every even integer $\mu>1$, no Hastings-McLeod-type solution of $\Ptwomu$ exists.
\end{theorem}
\begin{remark}
As discussed in Remark \ref{rem:reality}, the formal solution $y_f(x)$, for real $x<0$, may only be real-valued for odd $\mu$. For even values of $\mu$, the corresponding Hastings-Mcleod-type solution would not be real-valued and thus not mirror the distinctive properties of the original Hastings-McLeod solution.
\end{remark}

Hastings and McLeod originally tackled this problem while considering a different generalisation of $\Ptwo$, that is
\begin{equation*}
    y''(x)=2y(x)^{\alpha+1}+xy(x),
\end{equation*}
which coincides with $\Ptwomu$ when $\alpha=2$ and $\mu=1$ (i.e. $\Ptwo$). We will prove Theorem \ref{thrm:HM} by extending the approach taken by Hastings and McLeod to our case with $\alpha=2$ and more general $\mu$.

\subsection{Existence}

The proof of the existence (and non-existence when $\mu$ is even) part of Theorem \ref{thrm:HM} follows from a series of lemmas.

\begin{lemma}\label{lem:inteq}
    Given $\mu\in\mathbb{Z}^+$ and $k\in\mathbb{R}^+$, there exists a unique solution $y_k(x)$ of $\Ptwomu$ which is asymptotic to $kf(x)$ as $x\rightarrow+\infty$. At each $x$ for which the solution continues to exist as $x$ decreases, the solution and its derivatives are continuous functions of $k$. 
\end{lemma}
\begin{proof}
    By using $f(x)$ and $g(x)$ as integrating factors, we find that $y_k(x)$ must satisfy the integral equation
    \begin{equation}\label{eq:integral}
        y_k(x)=kf(x)+\int_x^\infty\left(f(x)g(t)-f(t)g(x)\right)y_k(t)^3dt,
    \end{equation}
    where $f$ and $g$ are defined in Equations (\ref{eq:genAi}). By application of a standard contraction mapping argument, in an appropriate space of decaying, continuous functions equipped with the uniform norm, Equation (\ref{eq:integral}) may be shown to have a unique solution, giving $y_k$ and its continuous dependence on $k$.
\end{proof}

\begin{lemma}\label{lem:evenmu}
    For every even $\mu\in\mathbb{Z}^+$, and every $k\in\mathbb{R}^+$, $y_k(x)$ remains positive as $x$ decreases and becomes singular at some finite value of $x$.
\end{lemma}
\begin{proof}
    Let $k\in\mathbb{R}^+$ and assume $\mu\in\mathbb{Z}^+$ is even. For sufficiently large $x\in\mathbb{R}^+$, $y_k(x)$ behaves like $kf(x)$. Therefore, there exists $x_0\in\mathbb{R}^+$ so that $y_k(x)>0$ and $y_k'(x)<0$ for all $x\in[x_0,\infty)$. Since $\mu$ is even, we immediately deduce from $\Ptwomu$ that $y_k(x)$ is convex and positive for all $x\in\mathbb{R}$.
    
    The remainder of the proof proceeds by contradiction. Assume that $y_k(x)$ remains bounded for all $x\in(-\infty,x_0)$, i.e., it is not singular in this interval.
    For all $x\in(-\infty,x_0)$, we have $y_k''(x)>y_k(x)^3$ and thus $y_k'(x)y_k''(x)<y_k'(x)y_k(x)^3$ (since $y_k'$ is negative). We integrate the latter autonomous inequality between $x$ and $x_0$, which after rearranging terms gives
    \begin{equation}\label{eq:theinequality}
        y_k'(x)^2>\frac{1}{2}y_k(x)^4+y_k'(x_0)^2-\frac{1}{2}y_k(x_0)^4.
    \end{equation}
    
    We now show that $y_k(x)$ is monotonically increasing to $+\infty$ as $x$ decreases. Suppose instead that $y_k'(x_1)$ vanishes at some $x_1<x_0$, while $y_k'(x)<0$ for all $x\in(x_1,x_0]$. The inequality \eqref{eq:theinequality} would then imply that $y_k(x_1)<y_k(x_0)$, which contradicts the fact that $y_k(x)$ is still increasing as we move from $x_0$ to $x_1$.
    
    Therefore, we can choose $x_1\in(-\infty,x_0)$ such that $\frac{1}{4}y_k(x_1)^4+y_k'(x_0)^2-\frac{1}{2}y_k(x_0)^4>0$, hence $y_k'(x)^2>\frac{1}{4}y_k(x)^4$ for all $x\in(-\infty,x_1)$. Therefore, $y_k'(x)<-\frac{1}{2}y_k(x)^2$, and another integration gives
    \begin{equation*}
        \frac{1}{y_k(x_1)}-\frac{1}{y_k(x)}>\frac{x_1-x}{2},
    \end{equation*}
    for all $x<x_1$, which contradicts the assumption that $y_k(x)$ is bounded and non-singular on the interval.
\end{proof}

As a direct result of Lemma \ref{lem:evenmu}, we have proven the non-existence of the Hastings-McLeod-type solution in the case of even $\mu$. We now continue our argument in the case of odd $\mu$.

\begin{lemma}\label{lem:oddmublowupopen}
    For every odd $\mu\in\mathbb{Z}^+$, the set $S_1\subset\mathbb{R}^+$ of $k$ values, for which $y_k(x)$ remains positive as $x$ decreases and becomes singular at some finite value of $x$, is an open set.
\end{lemma}
\begin{proof}
    Suppose we have $k_0\in\mathbb{R}^+$ and $x_0\in\mathbb{R}$ such that $y_{k_0}(x)$ blows up at $x=x_0$. We now show that there exists a neighbourhood of $k$ values (sufficiently close to $k_0$) such that the solution $y_k$ still becomes singular. Since $y_{k_0}(x_0)$ is unbounded, we may choose $x_1>x_0$ to make $y_{k_0}(x_1)$ arbitrarily large. We choose $x_1$ near $x_0$ such that
    \begin{equation}\label{ineq:a}
        y_{k_0}(x_1)^2>(|x_1|+1)^\mu, 
    \end{equation}
    noting then that $y_{k_0}(x)$ remains convex on $(x_1-1,x_1)$ while it continues to exist. Based on the continuous dependence of $y_k(x)$ on $k$, we may replace $k_0$ with sufficiently close $k$ so that the inequality (\ref{ineq:a}) continues to hold. Furthermore, on the interval $(x_1-1,x_1)$, we have $|x|<|x_1|+1$, i.e. $|x|^\mu<(|x_1|+1)^\mu$. Hence
    \begin{equation*}
        |x|^\mu<(|x_1|+1)^\mu<y_k(x_1)^2<y_k(x)^2.
    \end{equation*}
    So if $x<0$, we have
    \begin{equation}\label{ineq:b}
        y_k''(x)=2y_k(x)^3+x^\mu y_k(x)=2y_k(x)^3-|x|^\mu y_k(x)>2y_k(x)^3-y_k(x)^3=y_k(x)^3.
    \end{equation}
    Note that if $x\geq0$ we still have $y_k''(x)>y_k(x)^3$.
    
    We have therefore shown that $y_k''(x)>y_k(x)^3$ on $x\in(x_1-1,x_1)$. Then from the integration of this autonomous inequality, the proof of Lemma \ref{lem:evenmu} shows that $y_k(x)$ will become unbounded in this interval given sufficiently large $|y_k(x_1)|$ and $|y_k'(x_1)|$. As previously noted, we may choose $x_1$ so that these are arbitrarily large and so the desired result follows.
\end{proof}

\begin{lemma}\label{lem:oddmunegopen}
    For every odd $\mu\in\mathbb{Z}^+$, the set $S_2\subset\mathbb{R}^+$ of $k$ values, for which $y_k(x)$ becomes negative at some point as $x$ decreases, is an open set.
\end{lemma}
\begin{proof}
    Suppose there exist $k_0$ and $x_0$ such that $y_{k_0}(x_0)<0$, then by the continuity of $y_k$ in $k$, we may deduce that $y_k(x_0)<0$ for all $k$ in some neighborhood of $k_0$.
\end{proof}

\begin{lemma}
    For every odd $\mu\in\mathbb{Z}^+$, the set $S_1$ is non-empty.
\end{lemma}
\begin{proof}
    It follows from Equation (\ref{eq:integral}) that $y_k(x)>kf(x)$ for sufficiently large $x$. Differentiation of Equation (\ref{eq:integral}) shows similarly that $y_k'(x)<kf'(x)<0$. Using the same argument as in the proof of Lemma \ref{lem:oddmublowupopen}, we may choose sufficiently large $x_1$ and $k$ so that $y_k(x)$ blows up in $(x_1-1,x_1)$.
\end{proof}

\begin{lemma}
    For every odd $\mu\in\mathbb{Z}^+$, the set $S_2$ is non-empty.
\end{lemma}
\begin{proof}
    Let $\mu>1$ be an odd integer. Consider the graph of $p(x)$ given by the positive real root of  $2p(x)^2=-x^\mu$ on the interval $x\in(-\infty,0]$. Let $x_0<0$. Recall that $y_0(x)$ corresponds to the $k=0$ case in Lemma \ref{lem:inteq} and is the identically zero solution of $\Ptwomu$. By the continuity of $y_k(x_0)$ in $k$, we may choose $k$ sufficiently small so that $y_k(x_0)<p(x_0)$ and $p'(x_0)<y_k'(x_0)<0$. Examination of $\Ptwomu$ then shows that $y_k''(x)<0$ for $x\leq x_0$. Therefore, $y_k(x)$ vanishes in this interval and becomes negative to the left of $x_0$.
\end{proof}

We now conclude the proof of existence. We have two non-empty, open sets $S_1$ and $S_2$, which are disjoint. So there must exist at least one positive value of $k$, which lies in neither $S_1$ nor $S_2$. The solution $y_k(x)$ corresponding to such a $k$ is then finite and positive for all $x$. If such a solution were to decay to zero as $x\rightarrow-\infty$, then it would be oscillatory in this asymptotic limit (see Equation (\ref{eq:expOscillatory})), thus vanishing at some point and contradicting the fact that $k\notin S_2$. Therefore, we conclude that $y(x)\sim y_f(x)$ as $x\rightarrow-\infty$. So $y_k(x)$ is a Hastings-McLeod-type solution.

\subsection{Uniqueness}
We now give the argument for uniqueness to complete the proof of Theorem \ref{thrm:HM}.
\begin{lemma}
\label{lemma:k1bigk2}
    Let $\mu$ be an odd integer. Suppose $y_{k_1}(x)$ and $y_{k_2}(x)$ are Hastings-Mcleod-type solutions with $k_1>k_2$, then $y_{k_1}(x)>y_{k_2}(x)$ for all $x$.
\end{lemma}
\begin{proof}
  Note that the asymptotic behaviour as $x\to+\infty$ implies that $y_{k_1}(x)>y_{k_2}(x)$ for sufficiently large $x$. If the assertion of the lemma fails to hold, then there exists  $x_0\in \mathbb R$ such that $y_{k_1}(x_0)=y_{k_2}(x_0)$, while $y_{k_1}(x)>y_{k_2}(x)$ on $(x_0,\infty)$. For this to occur, we must have $y_{k_1}'(x_0)^2\leq y_{k_2}'(x_0)^2$.

  Define the energy function
    \begin{equation}
    \label{eqn:V}
        V_k(x)=y_k'(x)^2-x^\mu y_k(x)^2-y_k(x)^4.
    \end{equation}
    From $\Ptwomu$, it follows that
    \begin{equation}
    \label{eqn:VPrime}
        V_k'(x)=-\mu x^{\mu-1}y_k(x)^2.
    \end{equation}
    From Equation (\ref{eqn:V}), we have
    \begin{equation}\label{eqn:Vx0}
        V_{k_1}(x_0)-V_{k_2}(x_0)=y_{k_1}'(x_0)^2-y_{k_2}'(x_0)^2\leq 0,
    \end{equation}
    which implies $V_{k_1}(x_0)\leq V_{k_2}(x_0)$.

    Using Equation (\ref{eqn:VPrime}), $y_{k_1}(x)>y_{k_2}(x)$ implies that $V_{k_1}'(x)<V_{k_2}'(x)$. We proceed to integrate the autonomous inequality $V_{k_1}'(x)-V_{k_2}'(x)<0$ over the interval $(x_0,N)$, where $N\gg1$, giving
    \begin{align}\label{ineq:Nx0}
        V_{k_1}(N)-V_{k_2}(N)<V_{k_1}(x_0)-V_{k_2}(x_0).
    \end{align}
    
    However, for sufficiently large $N$ we have $V_{k_1}(N)-V_{k_2}(N)>0$, and we already deduced that $V_{k_1}(x_0)-V_{k_2}(x_0)\leq0$. These observations contradict the inequality (\ref{ineq:Nx0}). Therefore, $x_0$ cannot exist and the desired result follows.
\end{proof}
\begin{lemma}
Let $\mu$ be an odd integer. Suppose $y_{k_1}(x)$ and $y_{k_2}(x)$ are Hastings-Mcleod-type solutions. Then $k_1=k_2$.
\end{lemma}
\begin{proof}
  Suppose $k_1\not=k_2$. Without loss of generality, we assume $k_1>k_2$. As in the proof of Proposition \ref{prop:tritronquee}, we can appeal to another theorem of Olver \cite[Theorem 2.1, p.193]{olver1997asymptotics} to show that $p(x)=y_{k_1}(x)-y_{k_2}(x)$ satisfies
    \begin{align}\label{eqn:expsmall}
        p(x)=&b_1f^{-1/4}\exp\left(\int_{-\infty}^{x}f^{1/2}dx\right)\bigl(1+\textit{o}(1)\bigr)\nonumber\\
        &+b_2f^{-1/4}\exp\left(-\int_{-\infty}^{x}f^{1/2}dx\right)\bigl(1+\textit{o}(1)\bigr),
    \end{align}
    for some constants $b_1$ and $b_2$, where $f=-2\,x^\mu$, for sufficiently large $|x|$, $x<0$. Olver's theorem also states that the behaviour \eqref{eqn:expsmall} is twice differentiable. However, this contradicts the assumed asymptotic behaviour $y_{k_j}\sim y_f$, $k=1, 2$ given in \eqref{eq:yf} unless $b_1=0$.

    Suppose $b_2\not=0$. Then we use the differentiability of the result \eqref{eqn:expsmall} to obtain 
    \begin{align*}
        V_{k_1}(x)-V_{k_2}(x)&=-2\left(2y_f(x)^3+x^\mu\,y_f(x)\right)p(x)\bigl(1+\textit{o}(1)\bigr)\\
        &=-2y_f''(x)p(x)\bigl(1+\textit{o}(1)\bigr).
    \end{align*}
    For sufficiently large $|x|$, $x<0$, recalling that $\mu>1$ is an odd integer and $p(x)>0$, this asymptotic behaviour is negative. However, we showed in the previous lemma that $V_{k_1}(x)-V_{k_2}(x)$ is positive for sufficiently large positive $x$ and that $V_{k_1}'(x)-V_{k_2}'(x)<0$ for all $x$, which is a contradiction. Therefore, we must have $b_2=0$. 
    
    It follows that  $y_{k_1}(x)=y_{k_2}(x)$ for sufficiently large real $|x|$, with $x<0$. This contradicts the result of Lemma \ref{lemma:k1bigk2} unless $k_1=k_2$.
\end{proof}

This completes our proof of uniqueness of the Hastings-McLeod type solution.

\begin{remark} Consider the behaviours of $y_k(x)$ depending on the parameter $k\in\mathbb{R}^+$. For odd $\mu$, suppose $k^*$ is the unique value for which $y_{k*}(x)$ is the Hastings-Mcleod-type solution. Using the same argument as in Lemma \ref{lemma:k1bigk2}, we may show that $0<k<k^*$ implies $|y_k(x)|<y_{k^*}(x)$ for all $x$. So for $0<k<k^*$, we have a family of bounded solutions on the real-line; these solutions have oscillatory asymptotic behaviour given by Equation (\ref{eq:expOscillatory}). On the other hand, when $k^*<k$, we have solutions which are positive and convex for all $x$, becoming singular at some finite point as $x$ decreases.
\end{remark}

\section{Conclusion}\label{sec:conc}
In this paper, we investigated a family of nonlinear ODEs, parameterised by $\mu$, which are perturbations of the second Painlev\'{e} equation. We showed how such equations arise in an application, specifically, in a mathematical model of electro-diffusion. Our main results show that the perturbed equation possesses solutions that are natural generalisations of two celebrated solutions of $\Ptwo$.

We found that the perturbed equation always has solutions which we call tri\-tronque\'e, first described by Boutroux for $\Ptwo$. These solutions are asymptotic to a power-series expansion in a surprisingly broad annular region of the complex plane, near infinity. Furthermore, if our integer parameter $\mu$ is odd, $\Ptwomu$ has a solution analogous to the famous Hastings-McLeod solution of $\Ptwo$. The Hastings-McLeod-type solution is holomorphic, real-valued and positive on the entire real-line, with known asymptotic behaviour in either direction.

These results naturally give rise to similar questions about perturbations of other Painlev\'e equations. They suggest that certain behaviours, important in applications, are preserved by classes of perturbations.

\appendix
\section{Physical application of the perturbed equation}
\label{section:app}

This section shows how $\Ptwomu$ arises in Bass' electro-diffusion model \cite{bass1964irreversible} under a small change in the physical setting. Specifically, we allow particle flux to be non-constant.

Assume that $c_+(x)$ and $c_-(x)$ represent the respective (dimensionless) concentrations of two ionic species with equal and opposite charge, while $E(x)$ represents the (dimensionless) induced electric field. The following coupled nonlinear ODEs govern the model of electro-diffusion:
\begin{align}
    c'_+(x)&=E(x)c_+(x)+A_+,\label{eqn:ED1}\\
    c'_-(x)&=-E(x)c_-(x)+A_-,\label{eqn:ED2}\\
    \lambda^2E'(x)&=c_+(x)-c_-(x),\label{eqn:ED3}
\end{align}
where $\lambda$ is a constant, the terms $A_\pm$ arise due to Fick's law's standard passive diffusion process, and we note that $A_\pm$ is proportional to the ratio between the species flux in the $x$-direction, $\phi_\pm$, and the species diffusivity, $D_\pm$.

In the original model, $\phi_\pm$ is constant in $x$ due to the conservation of mass in a strictly one-dimensional setting. However, we consider the possibility that $\phi_\pm$ may undergo small amounts of variation, as the species may fluctuate in other spatial dimensions. However, we assume $D_\pm$ is constant.

Consider $A_\pm$ as a perturbative series, whereby $A_\pm$ is polynomial in $x$ and constant at leading order. For now, consider a linear correction term. Furthermore, we shall consider a somewhat simplified case whereby $A_+=A_-$, so we have
\begin{equation}
\label{eqn:purturbedParam}
    A_+=A_-=A_0+\epsilon x.
\end{equation}
Then by adding Equations (\ref{eqn:ED1}) and (\ref{eqn:ED2}), and applying Equation (\ref{eqn:ED3}), the result is
\begin{equation}
\label{eqn:ED4}
    c'_+(x)+c'_-(x)=\lambda^2E(x)E'(x)+2A_0+2\epsilon x.
\end{equation}
Upon integrating Equation (\ref{eqn:ED4}) and disregarding the constant of integration, we have
\begin{equation}
\label{eqn:ED5}
    c_+(x)+c_-(x)=\frac{\lambda^2}{2}E(x)^2+\epsilon x^2+2A_0x.
\end{equation}
Meanwhile, by differentiating Equation (\ref{eqn:ED3}), and then applying Equations (\ref{eqn:ED1}) and (\ref{eqn:ED2}), we obtain
\begin{equation}
\label{eqn:ED6}
    \lambda^2 E''(x)=E(x)(c_+(x)+c_-(x)).
\end{equation}
Hence, by combining Equations (\ref{eqn:ED5}) and (\ref{eqn:ED6}), we obtain
\begin{equation}
\label{eqn:EDFinal}
    \lambda^2 E''(x)=\frac{\lambda^2}{2}E(x)^3+(\epsilon x^2+2A_0x)E(x).
\end{equation}
Equation (\ref{eqn:EDFinal}) is a nonlinear ODE which governs the induced electric field as a function of $x$ and is similar in form to $\Ptwomu$ with $\mu=2$. Adding higher degree correction terms in Equation (\ref{eqn:purturbedParam}) would correspond to a higher degree polynomial $p(x)$ in the coefficient of $E$, that is,
\begin{equation*}
    \lambda^2 E''(x)=\frac{\lambda^2}{2}E(x)^3+p(x)E(x).
\end{equation*}
The presence of such higher degree ($>1$) terms in the coefficient function $p(x)$ motivates the study of our perturbed form of $\Ptwo{}$.

\printbibliography
\end{document}